\documentclass[11pt]{article}

\usepackage{tikz}
\usetikzlibrary{shapes,backgrounds}
\usepackage{tablefootnote}
\usepackage{graphicx}
\graphicspath{ {} }
\usepackage{enumitem}
\setlist{nosep}
\usepackage{color}
\usepackage{naaclhlt2018}
\usepackage{times}
\usepackage{url}
\usepackage{latexsym}
\usepackage{graphicx}
\usepackage{titlesec}
\usepackage{mdframed}
\usepackage{fancyvrb}
\usepackage{multirow}
\usepackage{footmisc}

\usepackage{comment}
\def\aclpaperid{327}
\usepackage{amsmath,mathtools,amssymb,amsthm}
\usepackage{array}
\newcolumntype{L}[1]{>{\raggedright\let\newline\\\arraybackslash\hspace{0pt}}m{#1}}
\newcolumntype{C}[1]{>{\centering\let\newline\\\arraybackslash\hspace{0pt}}m{#1}}
\newcolumntype{R}[1]{>{\raggedleft\let\newline\\\arraybackslash\hspace{0pt}}m{#1}}
\aclfinalcopy
\newtheorem{theorem}{Theorem}

\newtheorem{corollary}[theorem]{Corollary}
\newtheorem{definition}[theorem]{Definition}
\newtheorem{example}[theorem]{Example}

\newcommand{\RN}[1]{%
  \textup{\uppercase\expandafter{\romannumeral#1}}%
}

\providecommand*{\nat}[0]{\ensuremath{\mathbb N}}
\providecommand*{\natp}[0]{\ensuremath{\mathbb N_{\scriptscriptstyle
      +}}}
\newtheorem{innercustomgeneric}{\customgenericname}
\providecommand{\customgenericname}{}
\newcommand{\newcustomtheorem}[2]{%
  \newenvironment{#1}[1]
  {%
   \renewcommand\customgenericname{#2}%
   \renewcommand\theinnercustomgeneric{##1}%
   \innercustomgeneric
  }
  {\endinnercustomgeneric}
}

\newcustomtheorem{customthm}{Theorem}
\newcustomtheorem{customlemma}{Lemma}

\title{Recurrent Neural Networks as Weighted Language Recognizers}

\author{
Yining Chen\\
Dartmouth College\\
{\tt \small yining.chen.18@dartmouth.edu}\\\And
Sorcha Gilroy \\
ILCC\\
  University of Edinburgh \\
  {\tt \small s.gilroy@sms.ed.ac.uk} \\\And
  Andreas Maletti \\
  Institute of Computer Science\\
  Universit{\"a}t Leipzig \\
  {\tt \small andreas.maletti@uni-leipzig.de} \\\AND
  Jonathan May \\
  Information Sciences Institute\\
  University of Southern California\\
{\tt \small jonmay@isi.edu}\\\And 
  Kevin Knight \\
  Information Sciences Institute\\
  University of Southern California\\
  {\tt \small knight@isi.edu}\\}
\date{}

\allowdisplaybreaks

\begin{document}
\maketitle
\begin{abstract}
  We investigate the computational complexity of various problems for
  simple recurrent neural networks (RNNs) as formal models for 
  recognizing weighted languages.  We focus on the single-layer,
  ReLU-activation, rational-weight RNNs with softmax, which are
  commonly used in natural language processing applications.  We show
  that most problems for such RNNs are undecidable, including
  consistency, equivalence, minimization, and the determination of the
  highest-weighted string.  However, for consistent RNNs the last
  problem becomes decidable, although the solution length can surpass all computable bounds.  If additionally the string is limited to
  polynomial length, the problem becomes NP-complete and APX-hard. In summary, this shows that approximations and
  heuristic algorithms are necessary in practical applications of
  those RNNs.
\end{abstract}

\section{Introduction}
Recurrent neural networks (RNNs) are an attractive apparatus for
probabilistic language modeling~\cite{mikolov_zweig_2012}.  Recent
experiments show that RNNs significantly outperform other methods in
assigning high probability to held-out English text~\cite{limits}.

Roughly speaking, an RNN works as follows.  At each time step, it
consumes one input token, updates its hidden state vector, and
predicts the next token by generating a probability distribution over
all permissible tokens.  The probability of an input string is simply
obtained as the product of the predictions of the tokens constituting
the string followed by a terminating token.  In this manner, each RNN
defines a \emph{weighted language}; i.e. a total function from strings
to weights.  \newcite{siegelmann} showed that single-layer
rational-weight RNNs with saturated linear activation can compute any
computable function.  To this end, a specific architecture with
886~hidden units can simulate any Turing machine in real-time (i.e., 
each Turing machine step is simulated in a single time step).
However, their RNN encodes the whole input in its internal state,
performs the actual computation of the Turing machine when reading the
terminating token, and then encodes the output (provided an output is
produced) in a particular hidden unit.  In this way, their RNN allows
``thinking'' time (equivalent to the computation time of the Turing
machine) after the input has been encoded. 

We consider a different variant of RNNs that is commonly used in
natural language processing applications.  It uses ReLU activations,
consumes an input token at each time step, and produces softmax
predictions for the next token.  It thus immediately halts after
reading the last input token and the weight assigned to the input is
simply the product of the input token predictions in each step. 

Other formal models that are currently used to implement probabilistic
language models such as finite-state automata and context-free
grammars are by now well-understood.  A fair share of their utility
directly derives from their nice algorithmic properties.  For example,
the weighted languages computed by weighted finite-state automata are
closed under intersection (pointwise product) and union (pointwise
sum), and the corresponding unweighted languages are closed under
intersection, union, difference, and
complementation~\cite{droste_kuich_vogler_2013}.  Moreover, toolkits
like OpenFST~\cite{Allauzen2007} and
Carmel\footnote{https://www.isi.edu/licensed-sw/carmel/} implement
efficient algorithms on automata like minimization, intersection,
finding the highest-weighted path and the highest-weighted string. 

RNN practitioners naturally face many of these same problems.  For
example, an RNN-based machine translation system should extract the
highest-weighted output string (i.e., the most likely translation)
generated by an RNN, \cite{sutskever2014sequence,bahdanau2014neural}.
Currently this task is solved by approximation techniques like
heuristic greedy and beam searches. To facilitate the deployment of
large RNNs onto limited memory devices (like mobile phones)
minimization techniques would be beneficial. Again currently only
heuristic approaches like knowledge distillation \cite{rush16} are
available. Meanwhile, it is unclear whether we can determine if the computed
weighted language is consistent; i.e., if it is a probability
distribution on the set of all strings. Without a determination of the
overall probability mass assigned to all finite strings, a fair
comparison of language models with regard to perplexity is simply
impossible.

The goal of this paper is to study the above problems for the mentioned ReLU-variant of RNNs. More specifically, we ask and
answer the following questions:
\begin{itemize}
\item Consistency: Do RNNs compute consistent weighted languages?  Is
  the consistency of the computed weighted language decidable?
\item Highest-weighted string: Can we (efficiently) determine the
  highest-weighted string in a computed weighted language?
\item Equivalence: Can we decide whether two given RNNs compute the
  same weighted language?
\item Minimization: Can we minimize the number of neurons for a given
  RNN?
\end{itemize}

\section{Definitions and notations}
\label{definitions}
Before we introduce our RNN model formally, we recall some basic
notions and notation.  An \emph{alphabet}~$\Sigma$ is a finite set of
symbols, and we write~$\lvert \Sigma \rvert$ for the number of symbols
in~$\Sigma$.  A \emph{string}~$s$ over the alphabet~$\Sigma$ is a
finite sequence of zero or more symbols drawn from~$\Sigma$, and we
write~$\Sigma^*$ for the set of all strings over~$\Sigma$, of which
$\epsilon$~is the empty string.  The length of the string~$s \in
\Sigma^*$ is denoted~$\lvert s \rvert$ and coincides with the number
of symbols constituting the string.  As usual, we write~$A^B$ for the
set of functions $\{f \mid f \colon B \to A\}$.  A \emph{weighted
  language}~$L$ is a total function~$L \colon \Sigma^* \to \mathbb R$
from strings to real-valued weights.  For example, $L(a^n) =e^{-n}$ for all~$n \geq 0$ is such a weighted language.

We restrict the weights in our RNNs to the rational
numbers~$\mathbb{Q}$.  In addition, we reserve the use of a special 
symbol~$\$$ to mark the start and end of an input string.  To this
end, we assume that $\$ \notin \Sigma$ for all considered alphabets,
and we let $\Sigma_{\$} = \Sigma \cup \{\$\}$.

\begin{definition}
A single-layer RNN~$R$ is a $7$-tuple $\langle \Sigma, N, h_{-1}, W,
W', E, E' \rangle$, in which
\begin{itemize}
\item $\Sigma$~is an \emph{input alphabet}, 
\item $N$~is a finite set of \emph{neurons},
\item $h_{-1} \in \mathbb Q^N$ is an \emph{initial activation vector},
\item $W \in \mathbb Q^{N \times N}$~is a \emph{transition matrix},
\item $W' = (W'_a)_{a \in \Sigma_{\$}}$ is a $\Sigma_{\$}$-indexed
  family of \emph{bias vectors}~$W'_a \in \mathbb Q^N$,
\item $E \in \mathbb Q^{\Sigma_{\$} \times N}$ is a \emph{prediction
    matrix}, and
\item $E' \in \mathbb Q^{\Sigma_{\$}}$ is a \emph{prediction bias
    vector}.
\end{itemize}
\end{definition}

Next, let us define how such an RNN works.  We first prepare our
input encoding and the effect of our activation function.  For an
input string~$s = s_1 s_2 \cdots s_n \in \Sigma^*$ with~$s_1, \dotsc,
s_n \in \Sigma$, we encode this input as $\$s\$$ and thus assume that
$s_0 = \$$ and $s_{n+1} = \$$.  Our RNNs use ReLUs (Rectified Linear
Units), so for every~$v \in \mathbb Q^N$ we let~$\sigma \langle v
\rangle$ (the ReLU activation) be the vector $\sigma \langle v
\rangle \in \mathbb Q^N$ such that 
\[ \sigma \langle v\rangle(n) = \max \bigl(0, v(n) \bigr) \tag*{for
    every $n \in N$.} \]
In other words, the ReLUs act like identities on nonnegative inputs,
but clip negative inputs to~$0$.  We use softmax-predictions, so for
every vector~$p \in \mathbb Q^{\Sigma_{\$}}$ and $a \in \Sigma_{\$}$
we let
\[ \text{softmax} \langle p \rangle(a) = \frac{e^{p(a)}}{\sum_{a' \in
      \Sigma_{\$}} e^{p(a')}} \enspace. \] 
RNNs act in discrete time steps reading a single letter at each step.
We now define the semantics of our RNNs.

\begin{definition}
  Let $R = \langle \Sigma, N, h_{-1}, W, W', E, E'\rangle$ be an RNN,
  $s$~an input string of length~$n$ and $0 \leq t \leq n$ a time
  step.  We define
  \begin{itemize}
  \item the \emph{hidden state vector}~$h_{s, t} \in \mathbb Q^N$
    given by
    \[ h_{s,t} = \sigma \langle W \cdot h_{s, t-1} + W'_{s_t} \rangle
      \enspace, \]
    where $h_{s, -1} = h_{-1}$ and we use standard matrix product and
    point-wise vector addition,
  \item the \emph{next-token prediction vector}~$E_{s, t} \in \mathbb
    Q^{\Sigma_{\$}}$
    \[ E_{s,t} = E \cdot h_{s,t} + E' \]
  \item the \emph{next-token distribution}~$E'_{s, t} \in \mathbb
    R^{\Sigma_{\$}}$ 
    \[ E'_{s,t} = \text{\upshape softmax} \langle E_{s,t} \rangle
      \enspace. \]
  \end{itemize}
  Finally, the RNN~$R$ computes the weighted language~$R \colon
  \Sigma^* \to \mathbb R$, which is given for every input~$s = s_1
  \dotsm s_n$ as above by
  \[ R(s) = \prod_{t=0}^n E'_{s,t}(s_{t+1}) \enspace. \]
\end{definition}

In other words, each component~$h_{s, t}(n)$ of the hidden state
vector is the ReLU activation applied to a linear combination of all
the components of the previous hidden state vector~$h_{s, t-1}$
together with a summand~$W'_{s_t}$ that depends on the $t$-th input
letter~$s_t$.  Thus, we often specify~$h_{s, t}(n)$ as linear
combination instead of specifying the matrix~$W$ and the
vectors~$W'_a$.  The semantics is then obtained by predicting the
letters~$s_1, \dotsc, s_n$ of the input~$s$ and the final
terminator~$\$$ and multiplying the probabilities of the individual
predictions.

Let us illustrate these notions on an example.  We consider the
RNN~$\langle \Sigma, N, h_{-1}, W, W', E, E'\rangle$ with $\gamma \in
\mathbb Q$ and
\begin{itemize}
\item $\Sigma = \{a\}$ and $N = \{1, 2\}$,
\item $h_{-1} = (-1, 0)^T$ and
  \[ W =
    \begin{pmatrix}
      1 & 0 \\
      1 & 0 
    \end{pmatrix}
    \quad \text{and} \quad
    W'_{\$} = W'_a = 
    \begin{pmatrix}
      1 \\ 0
    \end{pmatrix}
  \]
\item $E(\$, \cdot) = (M + 1,\, -(M + 1))$ and $E(a, \cdot)
  = (1,\, -1)$ and 
\item $E'(\$) = -M$ and $E'(a) = 0$.
\end{itemize}
In this case, we obtain the linear combinations
\[ h_{s, t} = \sigma \Biggl\langle
  \begin{matrix}
    h_{s, t-1}(1) + 1 \\
    h_{s, t-1}(1)
  \end{matrix}
  \Biggr\rangle \]
computing the next hidden state components.  Given the initial
activation, we thus obtain $h_{s, t} = \sigma \langle t, t-1\rangle$.
Using this information, we obtain
\begin{align*}
  E_{s, t}({\$})
  &= (M+1) \cdot ( t - \sigma \langle t-1\rangle) - M \\ 
  E_{s, t}(a)
  &= t - \sigma \langle t-1\rangle \enspace.
\end{align*}
Consequently, we assign weight~$\tfrac{e^{-M}}{1+e^{-M}}$
to input~$\varepsilon$, weight~$\tfrac{1}{1+e^{-M}} \cdot 
\tfrac{e^1}{e^1 + e^1}$ to~$a$, and, more generally,
weight~$\tfrac{1}{1+e^{-M}} \cdot \tfrac 1{2^n}$ to~$a^n$.

Clearly the weight assigned by an RNN is always in the interval~$(0,1)$, which enables a probabilistic view. Similar to
weighted finite-state automata or weighted context-free grammars, each
RNN is a compact, finite representation of a weighted language. The softmax-operation enforces that the probability~$0$ is impossible as assigned weight, so each input string is principally possible.  In
practical language modeling, smoothing methods are used to change
distributions such that impossibility (probability~$0$) is removed.
Our RNNs avoid impossibility outright, so this can be considered a
feature instead of a disadvantage.

The hidden state~$h_{s,t}$ of an RNN can be used as scratch space for
computation.  For example, with a single neuron~$n$ we can count input
symbols in~$s$ via: 
\[ h_{s, t}(n) = \sigma \langle h_{s, t-1}(n) + 1\rangle \enspace. \]
Here the letter-dependent summand~$W'_a$ is universally~$1$.
Similarly, for an alphabet $\Sigma = \{a_1, \dotsc, a_m\}$ we can
use the method of \newcite{siegelmann} to encode the complete
input string~$s$ in base~$m+1$ using:
\[ h_{s,t}(n) = \sigma \langle (m+1) h_{s, t-1}(n) + c(s_t) \rangle
  \enspace, \] 
where $c \colon \Sigma_{\$} \to \{0, \dotsc, m\}$ is a bijection.  In
principle, we can thus store the entire input string (of unbounded
length) in the hidden state value~$h_{s,t}(n)$, but our RNN model
outputs weights at each step and terminates immediately once the final
delimiter~$\$$ is read.  It must assign a probability to a string
\emph{incrementally} using the chain rule decomposition $p(s_1
\dotsm s_n) = p(s_1) \cdot \ldots \cdot p(s_n \mid s_1 \dotsm
s_{n-1})$.

\begin{figure*}[t]
  \begin{center}
    \begin{tabular}{|l|c|c|c|} \hline
      & $R_1(a^n) = 2^{-(n+1)}$
      & $R_2(\varepsilon) \approx 0$% = 1 - \frac{1}{(1+e^{-\gamma})}$
      & $R_3(a^{100}) \approx 1$\\
      && $R_2(a^n) \approx 2^{-n} \; (n \geq 1)$%\frac{1}{(1+e^{-\gamma})} \cdot 2^{-n}$
      & $R_3(a^n) \approx 0 \; (n \neq 100)$ \\ \hline
      $N$ & $\{1\}$ & $\{1,2\}$ & $\{1,2,3\}$ \\ \hline
      $h_{-1}$
      & $\begin{pmatrix} 0 \end{pmatrix}$
      & $\begin{pmatrix} -1 \\ 0 \end{pmatrix}$
      & $\begin{pmatrix} 0 \\ 0 \\ 0 \end{pmatrix}$ \\ \hline
      $W$
      & $\begin{pmatrix} 0 \end{pmatrix}$
      & $\begin{pmatrix} 1 & 0 \\ 1 & 0 \end{pmatrix}$
      & $\begin{pmatrix} 0 & 0 & 1 \\ 0 & 0 & 1 \\ 0 & 0 &
        1 \end{pmatrix}$ \\ \hline
      $W'_{\$} \quad W'_a$
      & $\begin{pmatrix} 0 \end{pmatrix} \quad
        \begin{pmatrix} 0 \end{pmatrix}$
      & $\begin{pmatrix} 1 \\ 0 \end{pmatrix} \quad
      \begin{pmatrix} 1 \\ 0 \end{pmatrix}$
      & $\begin{pmatrix} -99 \\ -100 \\ 1 \end{pmatrix} \quad
      \begin{pmatrix} -99 \\ -100 \\ 1 \end{pmatrix}$ \\ \hline
      $E_{\$} \quad E_a$
      & $\begin{pmatrix} 0 \end{pmatrix} \quad
        \begin{pmatrix} 0 \end{pmatrix}$
      & $\begin{pmatrix}  M +1 \\ -(M + 1) \end{pmatrix} \quad \begin{pmatrix} 1 \\ -1 \end{pmatrix}$
     % & $\begin{pmatrix} \gamma+1 \\ -(\gamma+1) \end{pmatrix} \quad
      %  \begin{pmatrix} 1 \\ -1 \end{pmatrix}$
      & $\begin{pmatrix} M \\ -M \\ 0 \end{pmatrix} \quad
        \begin{pmatrix} -M \\ M \\ 0 \end{pmatrix}$ \\
      \hline
      $E'_{\$} \quad E'_a$
      & $0 \quad 0$
      & $-M \quad 0$
      & $-M \quad 0$ \\ \hline
      & \includegraphics[width=0.27\textwidth]{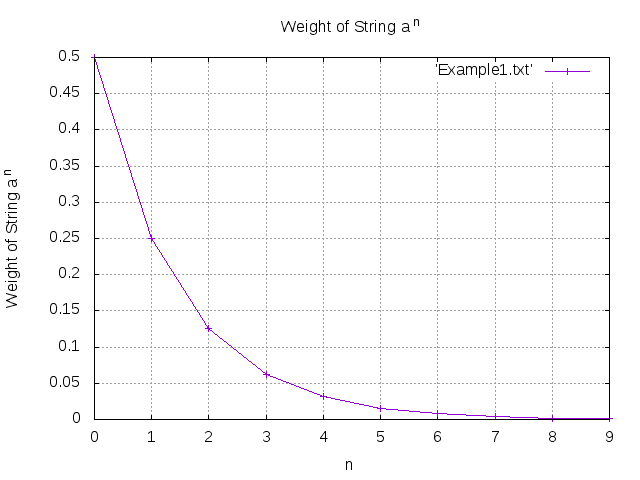}
      & \includegraphics[width=0.27\textwidth]{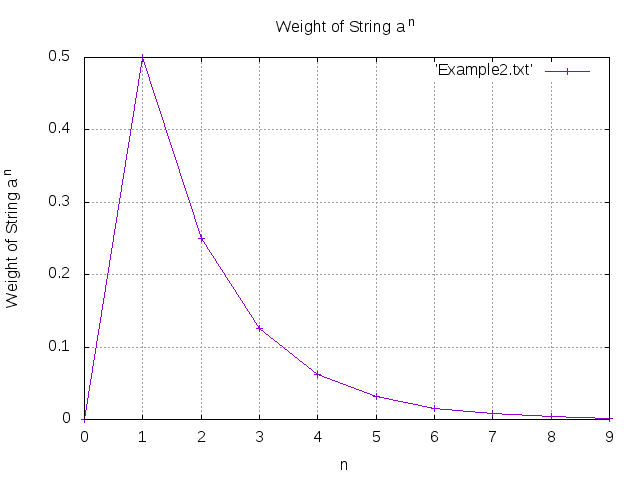}
      & \includegraphics[width=0.27\textwidth]{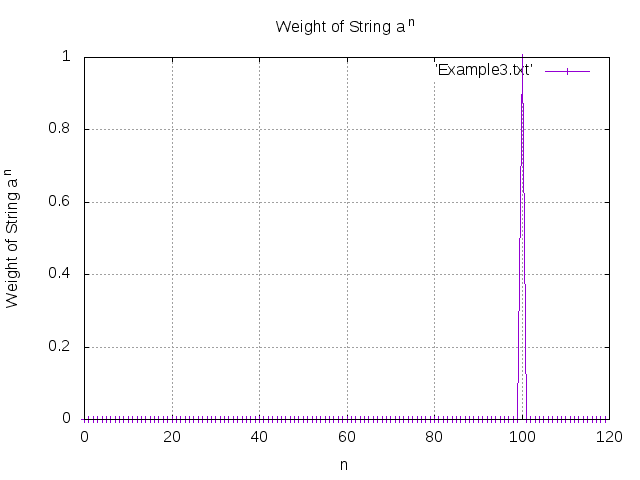} \\ \hline
    \end{tabular}
  \end{center}
  \caption{Sample RNNs over single-letter alphabets, and the weighted languages they recognize. $M$ is some positive rational number which depends on the desired error margin. If we want to express the second and the third languages with error margin $\delta$, $M$ is chosen so that $M > -\ln\frac{\delta}{1-\delta}$ in column $2$, and chosen so that $(1+e^{-M})^{100} < \frac{1}{1-\delta}$ in column $3$.}
  \label{examples}
\end{figure*}

Let us illustrate our notion of RNNs on some additional examples.
They all use the alphabet~$\Sigma = \{a\}$ and are illustrated and
formally specified in Figure~\ref{examples}.  The first column shows
an RNN~$R_1$ that assigns~$R_1(a^n) = 2^{-(n+1)}$.  The
next-token prediction matrix ensures equal values for $a$~and~$\$$ at
every time step.  The second column shows the RNN~$R_2$, which we
already discussed.  In the beginning, it heavily biases the next
symbol prediction towards~$a$, but counters it starting at~$t=1$.  The
third RNN~$R_3$ uses another counting mechanism with $h_{s,t} =
\sigma \langle t-100, t-101, t\rangle$.  The first two components are
ReLU-thresholded to zero until~$t > 101$, at which point they
overwhelm the bias towards~$a$ turning all future predictions to~$\$$. 

\section{Consistency}
\label{consistency}
We first investigate the consistency problem for an RNN~$R$, which
asks whether the recognized weighted language~$R$ is indeed a
probability distribution.  Consequently, an RNN~$R$ is
\emph{consistent} if $\sum_{s \in \Sigma^*} R(s) = 1$.  We first show
that there is an inconsistent RNN, which together with our examples
shows that consistency is a nontrivial property of
RNNs.\footnote{\label{footnotepcfg} For comparison, all probabilistic
  finite-state automata are consistent, provided no transitions exit
  final states.  Not all probabilistic context-free grammars are
  consistent; necessary and sufficient conditions for consistency are
  given by \newcite{booth_thompson_1973}.  However, probabilistic
  context-free grammars obtained by training on a finite corpus using
  popular methods (such as expectation-maximization) are guaranteed to
  be consistent~\cite{nederhof_satta_2006}.}

We immediately use a slightly more complex example, which we will
later reuse.

\begin{example} \upshape
  \label{ex:consistent}
  Let us consider an arbitrary RNN
  \[ R = \langle \Sigma, N, h_{-1}, W, W', E, E'\rangle \]
  with the single-letter alphabet~$\Sigma = \{a\}$, the neurons~$\{1, 2,
  3, n, n'\} \subseteq N$, initial activation $h_{-1}(i) = 0$ for all $i
  \in \{1, 2, 3, n, n'\}$, and the following linear combinations:
  \begin{align*}
    h_{s, t}(1)
    &= \sigma \langle h_{s, t-1}(1) + h_{s, t-1}(n) - h_{s, t-1}(n')
      \rangle \\
    h_{s, t}(2)
    &= \sigma \langle h_{s, t-1}(2) + 1\rangle \\
    h_{s, t}(3)
    &= \sigma \langle h_{s, t-1}(3) + 3 h_{s, t-1}(1) \rangle \\[1ex]
    E_{s, t}(\$)
    &= h_{s, t}(3) - h_{s, t}(2) \\
    E_{s, t}(a)
    &= h_{s, t}(2)
  \end{align*}
  Now we distinguish two cases:  \\
  \textbf{Case~1:} If $h_{s, t}(n) - h_{s, t}(n') = 0$ for all $t \in
  \nat$, then $h_{s, t}(1) = 0$ and $h_{s, t}(2) = t + 1$ and $h_{s,
    t}(3) = 0$.  Hence we have $E_{s, t}(\$) = -(t+1)$ and $E_{s,
    t}(a) = t+1$.  In this case the termination probability
  \[ E'_{s, t}(\$) = \frac{e^{-(t+1)}}{e^{-(t+1)} + e^{t+1}} = \frac
    1{1+e^{2(t+1)}} \]
  (i.e., the likelihood of predicting~$\$$) shrinks rapidly
  towards~$0$, so the RNN assigns less than 15\% of the probability
  mass to the terminating sequences (i.e., the finite strings), so the
  RNN is inconsistent (see Lemma 15
%  ~\ref{lm:incon} 
  in the
  appendix). \\[1ex]
  \textbf{Case 2:} Suppose that there exists a time point~$T \in \nat$
  such that for all $t \in \nat$
  \[ h_{s, t}(n) - h_{s, t}(n') =
    \begin{cases}
      1 & \text{if } t = T \\
      0 & \text{otherwise.}
    \end{cases} \]
  Then $h_{s, t}(1) = 0$ for all~$t \leq T$ and $h_{s, t}(1) = 1$
  otherwise.  In addition, we have $h_{s, t}(2) = t + 1$ and $h_{s,
    t}(3) = \sigma\langle 3(t-T-1)\rangle$.  Hence we have
  \begin{align*}
    E_{s, t}(\$)
    &= \sigma\langle 3(t-T-1)\rangle - (t+1) \\
    &=
      \begin{cases}
        -(t+1) & \text{if } t \leq T \\
        2t - 3T - 4 & \text{otherwise}
      \end{cases} \\
    E_{s, t}(a)
    &= t + 1 \enspace,
  \end{align*}
  which shows that the probability
  \[ E'_{s, t}(\$) =
    \begin{cases}
      \frac{1}{1+e^{2(t+1)}} & \text{ if } t \leq T \\[1ex]
      \frac{e^{t-3T-5}}{1+e^{t-3T-5}} & \text{otherwise}
    \end{cases}
  \]
  of predicting~$\$$ increases over time and eventually (for $t \gg
  3T$) far outweighs the probability of predicting~$a$.  Consequently,
  in this case the RNN is consistent (see Lemma 16
%  ~\ref{lm:con} 
in the appendix).
\end{example}

We have seen in the previous example that consistency is not trivial
for RNNs, which takes us to the consistency problem for RNNs: 

\paragraph{Consistency:} Given an RNN~$R$, return ``yes'' if $R$~is
consistent and ``no'' otherwise.

\medskip
We recall the following theorem, which, combined with our
example, will prove that consistency is unfortunately undecidable for
RNNs.

\begin{theorem}[Theorem~2 of \newcite{siegelmann}]
  \label{lm-TuringComplete}
  Let $M$~be an arbitrary deterministic Turing machine.  There exists
  an RNN
  \[ R = \langle \Sigma, N, h_{-1}, W, W', E, E'\rangle \]
  with saturated linear activation, input alphabet~$\Sigma = \{a\}$,
  and $1$~designated neuron~$n \in N$ such that for all $s \in
  \Sigma^*$ and $0 \leq t \leq \lvert s\rvert$
  \begin{itemize}
  \item $h_{s, t}(n) = 0$ if $M$~does not halt on~$\varepsilon$, and
  \item if $M$~does halt on empty input after $T$~steps, then
    \[ h_{s, t}(n) =
      \begin{cases}
        1 & \text{if } t = T \\
        0 & \text{otherwise.}
      \end{cases}
    \]
  \end{itemize}
\end{theorem}

In other words, such RNNs with saturated linear activation can
semi-decide halting of an arbitrary Turing machine in the sense that a
particular neuron achieves value~$1$ at some point during the
evolution if and only if the Turing machine halts on empty input. 
An RNN with saturated linear activation is an RNN following our
definition with the only difference that instead of our
ReLU-activation~$\sigma$ the following saturated linear
activation~$\sigma' \colon \mathbb Q^N \to \mathbb Q^N$ is used.  For
every vector~$v \in \mathbb Q^N$ and $n \in N$, let
\[ \sigma'\langle v\rangle(n) =
  \begin{cases}
    0 & \text{if } v(n) < 0 \\
    v(n) & \text{if } 0 \leq v(n) \leq 1 \\
    1 & \text{if } v(n) > 1 \enspace.
  \end{cases}
\]

%\marginpar{\textcolor{red}{show the
%  function plots for these functions here to 
%  convince the reader}}

Since $\sigma'\langle v\rangle = \sigma\langle v\rangle -
\sigma\langle v- \vec 1\rangle$ for all $v \in \mathbb Q^N$, and the
right-hand side is a linear transformation, we can easily simulate
saturated linear activation in our RNNs.  To this end, each neuron~$n 
\in N$ of the original RNN~$R = \langle \Sigma, N, h_{-1}, U, U', E,
E'\rangle$ is replaced by two neurons $n_1$~and~$n_2$ in the new
RNN~$R' = \langle \Sigma, N', h'_{-1}, V, V', F, F'\rangle$ such that
$h_{s, t}(n) = h'_{s, t}(n_1) - h'_{s, t}(n_2)$ for all $s \in
\Sigma^*$ and $0 \leq t \leq \lvert s\rvert$, where the evaluation
of~$h'_{s,t}$ is performed in the RNN~$R'$.  More precisely, we
use the transition matrix~$V$ and bias function~$V'$, which is
given by 
\begin{align*}
  V(n_1, n'_1)
  &= V(n_2, n'_1) = U(n, n') \\
  V(n_1, n'_2) 
  &= V(n_2, n'_2) = -U(n, n') \\
  V'_a(n_1)
  &= U'_a(n) \\
  V'_a(n_2)
  &= U'_a(n) - 1 \\
  h'_{-1}(n_1)
  &= h_{-1}(n) \\
  h'_{-1}(n_2)
  &= 0
\end{align*}
for all $n, n' \in N$ and $a \in \Sigma \cup \{\$\}$, where
$n_1$~and~$n_2$ are the two neurons corresponding to~$n$ and
$n'_1$~and~$n'_2$ are the two neurons corresponding to~$n'$ (see
Lemma 17
%~\ref{lm:sim}
in the appendix).

\begin{corollary}
  \label{cor:sim}
  Let $M$~be an arbitrary deterministic Turing machine.  There exists
  an RNN
  \[ R = \langle \Sigma, N, h_{-1}, W, W', E, E'\rangle \]
  with input alphabet~$\Sigma = \{a\}$ and $2$~designated
  neurons~$n_1, n_2 \in N$ such that for all $s \in \Sigma^*$ and $0
  \leq t \leq \lvert s\rvert$
  \begin{itemize}
  \item $h_{s, t}(n_1) - h_{s, t}(n_2) = 0$ if $M$~does not halt
    on~$\varepsilon$, and 
  \item if $M$~does halt on empty input after $T$~steps, then
    \[ h_{s, t}(n_1) - h_{s, t}(n_2) =
      \begin{cases}
        1 & \text{if } t = T \\
        0 & \text{otherwise.}
      \end{cases}
    \]
  \end{itemize}
\end{corollary}

We can now use this corollary together with the RNN~$R$ of
Example~\ref{ex:consistent} to show that the consistency problem is
undecidable.  To this end, we simulate a given Turing machine~$M$ and
identify the two designated neurons of Corollary~\ref{cor:sim} as 
$n$~and~$n'$ in Example~\ref{ex:consistent}.  It follows that
$M$~halts if and only if $R$~is consistent.  Hence we reduced the
undecidable halting problem to the consistency problem, which shows
the undecidability of the consistency problem.

\begin{theorem}
  \label{thm:consistent}
  The consistency problem for RNNs is undecidable.
\end{theorem}

As mentioned in Footnote~\ref{footnotepcfg}, probabilistic
context-free grammars obtained after training on a finite corpus using
the most popular methods are guaranteed to be consistent.  At least
for 2-layer RNNs this does not hold.

\begin{theorem}
  \label{backprop_inconsistent}
  A two-layer RNN trained to a local optimum using
  Back-propagation-through-time~(BPTT) on a finite corpus is not
  necessarily consistent.
\end{theorem}

\begin{proof}
  The first layer of the RNN~$R$ with a single alphabet symbol~$a$
  uses one neuron~$n'$ and has the following behavior:
  \begin{align*}
    h_{-1}(n')
    &= 0 \\
    h_{s, t}(n')
    &= \sigma \langle h_{s, t-1}(n') + 1 \rangle
  \end{align*}
  The second layer uses neuron~$n$ and takes~$h_{s,t}(n')$ as input at
  time~$t$: 
  \begin{align*}
    h_{s,t}(n)
    &= \sigma \langle h_{s,t}(n') - 2 \rangle \\
    E_{s,t}(a)
    &= h_{s,t}(n) \qquad
    & E_{s, t}(\$) &= 0 \\
    E'_{s,t}(a)
    &= \begin{cases}
      \frac{1}{2}
      & \text{if } t \leq 1 \\
      \frac{e^{(t-1)}}{1+e^{(t-1)}}
      & \text{otherwise.}
    \end{cases}
  \end{align*}
  Let the training data be~$\{a\}$.  Then the objective we wish to 
  maximize is simply~$R(a)$.  The derivative of this objective with
  respect to each parameter is~$0$, so applying gradient descent
  updates does not change any of the parameters and we have converged
  to an inconsistent RNN.
\end{proof}

It remains an open question whether there is a single-layer RNN that
also exhibits this behavior.

\section{Highest-weighted string}
Given a function~$f \colon \Sigma^* \to \mathbb R$ we are often
interested in the highest-weighted string.  This corresponds to the
most likely sentence in a language model or the most likely
translation for a decoder RNN in machine translation.

For deterministic probabilistic finite-state automata or context-free
grammars only one path or derivation exists for any given string, so
the identification of the highest-weighted string is the same task as
the identification of the most probable path or derivation.  However,
for nondeterministic devices, the highest-weighted string is often
harder to identify, since the weight of a string is the sum of the
probabilities of all possible paths or derivations for that string.  A
comparison of the difficulty of identifying the most probable
derivation and the highest-weighted string for various models is
summarized in Table~\ref{ComparisonMPS}, in which we marked our
results in bold face.

\begin{table}
  \begin{center}
    \begin{tabular}{|c|c|c|} 
      \hline
      & Best-path
      & Best-string \\ 
      \hline
      General RNN
      & \multicolumn{2}{c|}{\textbf{Undecidable}} \\
      Consistent RNN
      & \multicolumn{2}{c|}{\textbf{NP-c}
        \tablefootnote{Restricted to solutions of polynomial length}}
      \\ \hline
      Det.\@ PFSA/PCFG
      & \multicolumn{2}{c|}{P
        \tablefootnote{Dijkstra shortest path / \cite{knuth_1977}}} \\ \cline{3-3}
      Nondet.\@ PFSA/PCFG
      & & \multicolumn{1}{|c|}{NP-c
          \tablefootnote{\cite{casacuberta_higuera_2000} /
          \cite{simaan_1996}}} \\
      \hline
    \end{tabular}
  \end{center}
  \caption{Comparison of the difficulty of identifying the most
    probable derivation (Best-path) and the highest-weighted string
    (Best-string) for various models.}
  \label{ComparisonMPS}
\end{table}

We present various results concerning the difficulty of identifying
the highest-weighted string in a weighted language computed by an
RNN.  We also summarize some available algorithms.  We start with the
formal presentation of the three studied problems. 
\begin{enumerate}
\item \textbf{Best string:} Given an RNN~$R$ and $c \in (0,1)$, does
  there exist $s \in \Sigma^*$ with $R(s) > c$?
\item \textbf{Consistent best string:} Given a consistent RNN~$R$ and
  $c \in (0,1)$, does there exist $s \in \Sigma^*$ with $R(s) > c$?
\item \textbf{Consistent best string of polynomial length:} Given a
  consistent RNN~$R$, polynomial~$\mathcal{P}$ with $\mathcal{P}(x) \ge x$ for $x \in \mathbb{N}^+$, and $c \in (0,1)$, does there
  exist $s \in \Sigma^*$ with $\lvert s\rvert \leq \mathcal{P}(\lvert R\rvert)$
  and $R(s) > c$?
\end{enumerate}
As usual the corresponding optimization problems are not significantly
simpler than these decision problems.  Unfortunately, the general
problem is also undecidable, which can easily be shown using our
example.

\begin{theorem}
  \label{tm-MPS}
  The best string problem for RNNs is undecidable.
\end{theorem}
\begin{proof}
  Let $M$ be an arbitrary Turing machine and again consider the
  RNN~$R$ of Example~\ref{ex:consistent} with the neurons $n$~and~$n'$
  identified with the designated neurons of Corollary~\ref{cor:sim}.
  We note that $R(\varepsilon) = \tfrac 1{1+e^2} < 0.12$ in both
  cases.  If $M$~does not halt, then $R(a^n) \leq \tfrac
  1{1+e^{2(n+1)}} \leq \tfrac 1{1+e^2} < 0.12$ for all~$n \in \nat$.
  On the other hand, if $M$~halts after $T$~steps, then
  \begin{align*}
    &\phantom{{}={}} R(a^{3T-5}) \\
    &= \Bigl(\prod_{t = 0}^T \frac{e^{2(t+1)}}{1 + e^{2(t+1)}} \Bigr)
      \cdot \Bigl(\prod_{t = T+1}^{3T-6} \frac{1}{1+e^{t-3T-5}} \Bigr)
      \cdot \frac12 \\
    &\geq \frac{2}{(-1, e^{-2})_\infty} \cdot \Bigl(\prod_{t =
      T+1}^{3T-6} \frac{e^{3T+5-t}}{e^{3T+5-t+1}} \Bigr) \cdot \frac12
    \\
    &\geq \frac{2}{(-1, e^{-2})_\infty \cdot (-1, e^{-1})_\infty} \geq 0.25
  \end{align*}
  using Lemma 14
  %~\ref{lm:help}
  in the appendix.  Consequently, a
  string with weight above~$0.12$ exists if and only if $M$ halts, so
  the best string problem is also undecidable.
\end{proof}

%\begin{center}
%  \includegraphics[width=0.5\textwidth]{HWS}
%\end{center}

If we restrict the RNNs to be consistent, then we can easily decide
the best string problem by simple enumeration.

\begin{theorem}
  The consistent best string problem for RNNs is decidable.
\end{theorem}

\begin{proof}
  Let $R$ be the RNN over alphabet~$\Sigma$ and $c \in (0, 1)$ be the
  bound.  Since $\Sigma^*$~is countable, we can enumerate it via~$f
  \colon \nat \to \Sigma^*$.  In the algorithm we compute $S_n =
  \sum_{i = 0}^n R(f(i))$ for increasing values of~$n$.  If we
  encounter a weight~$R(f(n)) > c$, then we stop with answer ``yes.''
  Otherwise we continue until $S_n > 1 - c$, at which point we stop
  with answer ``no.''

  Since $R$~is consistent, $\lim_{i\to\infty} S_i = 1$, so this
  algorithm is guaranteed to terminate and it obviously decides the
  problem.
\end{proof}

Next, we investigate the length~$\lvert w_R^{\text{max}} \rvert$ of the
shortest string~$w_R^{\text{max}}$ of maximal weight in the weighted
language~$R$ generated by a consistent RNN~$R$ in terms of its (binary
storage) size~$\lvert R\rvert$.  As already mentioned
by~\newcite{siegelmann} and evidenced here, only small precision
rational numbers are needed in our constructions, so we assume that 
$\lvert R \rvert \leq c \cdot \lvert N \rvert^2$ for a (reasonably
small) constant~$c$, where $N$~is the set of neurons of~$R$.  We show
that no computable bound on the length of the best string can exist,
so its length can surpass all reasonable bounds.

\begin{theorem}
  Let $f \colon \natp \to \nat$ be the function with 
  \[  f(n) = \max_{\substack{\text{consistent RNN~$R$} \\
        \lvert R \rvert \leq n}} \lvert w_R^{\text{max}} \rvert \]
  for all $n \in \natp$.  There exists no computable function $g
  \colon \nat \to \nat$ with $g(n) \geq f(n)$ for all $n \in
  \nat$. 
\end{theorem}
%\iffalse
\begin{proof}
%The proof is given in the supplementary material. We use the ``Busy Beaver'' number \citep{rad62} %to show that $g$ is not computable.
%\iffalse
  In the previous section (before Theorem~\ref{thm:consistent}) we
  presented an RNN~$R_M$ that simulates an arbitrary
  (single-track) Turing machine~$M$ with $n$~states.
  By~\newcite{siegelmann} we have $\lvert R_M \rvert  \leq c \cdot (4n
  + 16)$.  Moreover, we observed that this RNN~$R_M$ is
  consistent if and only if the Turing machine~$M$ halts on empty
  input.  In the proof of Theorem~\ref{tm-MPS} we have additionally
  seen that the length~$\lvert w_R^{\text{max}} \rvert$ of its best
  string exceeds the number~$T_M$ of steps required to halt.

  For every $n \in \nat$, let $BB(n)$ be the $n$-th ``Busy Beaver''
  number~\cite{rad62}, which is
  \[ BB(n) = \max_{\substack{\text{normalized $n$-state Turing
          machine~$M$ with} \\ \text{2 tape symbols that halts on
          empty input}}} T_M \] 
  It is well-known that $BB \colon \natp \to \nat$ cannot be bounded
  by any computable function.  However,
  \begin{align*}
    BB(n)
    &\leq \max_{\substack{\text{normalized $n$-state Turing
      machine~$M$ with} \\ \text{and 2 tape symbols that halts on
    empty input}}} \lvert w_{R_M}^{\text{max}} \rvert \\
    &\leq \max_{\substack{\text{consistent RNN~$R$} \\ \lvert R \rvert
    \leq c \cdot (4n + 16)}} \lvert w_R^{\text{max}} \rvert \\
    &= f(4nc + 16c) \enspace,
  \end{align*}
  so $f$~clearly cannot be computable and no computable function~$g$
  can provide bounds for~$f$. 
 % \fi
\end{proof}

Finally, we investigate the difficulty of the best string problem for consistent RNN restricted to solutions of polynomial length.

\begin{theorem}
Identifying the best string of polynomial length in a consistent RNN is NP-complete and APX-hard.
\end{theorem}
{\em Proof sketch.} Clearly, we can guess an input string of polynomial length, run the RNN, and verify whether its weight exceeds the given bound in polynomial time.  Therefore the problem is trivially in NP.  For NP-hardness, we reduce from the 0-1 Integer Linear Programming Feasibility Problem:

\paragraph{0-1 Integer Linear Programming Feasibility:} Given: $n$ variables $x_1, x_2, ..., x_n$ which can only take values in $\{0, 1\}$, and $k$ constraints ($k \ge 1$): $\sum_{j=1}^{n}{A_{ij}x_j} - B_i \le 0$, $\forall i \in [k]$. $A \in \mathbb{Z}^{k \times n}$, $B \in \mathbb{Z}^k$. Return: Yes iff. there is a feasible solution $x=(x_1, x_2, ..., x_n) \in \{0, 1\}^n$ that satisfies all $k$ constraints.
\medskip

Suppose we are given an instance of the above problem. We construct an instance of the consistent best string of polynomial length problem with input $\langle R, c\rangle$. Our construction ensures that the only length at which a string can have weight greater than $c$ is $n$. Thus, if there is any string whose weight is greater than $c$, the given instance of 0-1 Integer Linear Programming Problem is feasible; otherwise it is not.

Our reduction is a Polynomial-Time Approximation Scheme (PTAS) reduction and preserves approximability. Since 0-1 Integer Linear Programming Feasibility is NP-complete and the corresponding maximization problem is APX-complete, consistent best string of polynomial length is NP-complete and APX-hard, meaning there is no PTAS to find the best string bounded by polynomial length (i.e. the best we can hope for in polynomial time is a constant-factor approximation algorithm) unless P $=$ NP.

The full proof is given in the appendix.

If we assume that the solution length is bounded by some finite number, we can convert algorithms from \newcite{Higuera2013ComputingTM} for computing the most probable string in PFSAs for use in RNNs. Such algorithms would be similar to beam search \cite{Lowerre:1976:HSR:907741} used most widely in practice.

%\subsection{Randomized and Deterministic Algorithms for \textbf{MPS-CST}}
%Suppose the solution length is bounded by $b$ (since in practice the mostly likely string is usually reasonably short), and suppose calculating the forward pass of an RNN with $|N|$ hidden neurons for one time step takes time $O(|N|+|\Sigma|)$, we can directly convert the algorithm in \newcite{Higuera2013ComputingTM} for most probable string in PFSA to get an $O(\frac{b(|N|+|\Sigma|)}{p})$ randomized algorithm based on sampling and an $O(\frac{b(|N|+|\Sigma|)|\Sigma|}{p})$ deterministic algorithm based on keeping the $\frac{1}{p}$-best prefixes (akin to beam search \cite{Lowerre:1976:HSR:907741} used most widely in practice).

%We can also use their sampling method to find b: $\forall p>0, \delta >0$, if we draw a sample $S$ of size at least $\frac{1}{p}\ln{\frac{1}{\delta}}$, the following holds with probability at least $1-\delta$: the probability of sampling a string $x$ longer than any string we have seen in $S$ is less than $p$.

\section{Equivalence}

We prove that equivalence of two RNNs is undecidable. For comparison, equivalence of two deterministic WFSAs can be tested in time $O(|\Sigma|(|Q_A|+|Q_B|)^3)$, where $|Q_A|$, $|Q_B|$ are the number of states of the two WFSAs and $|\Sigma|$ is the size of the alphabet \cite{cortes_mohri_rastogi_2007}; equivalence of nondeterministic WFSAs are undecidable \cite{griffiths_1968}. The decidability of language equivalence for deterministic probabilistic push-downtown automata (PPDA) is still open \cite{forejt_jančar_kiefer_worrell_2014}, although equivalence for deterministic unweighted push-downtown automata (PDA) is decidable \cite{Sénizergues1997}.

The equivalence problem is formulated as follows:
\paragraph{Equivalence:} Given two RNNs $R$ and $R'$, return ``yes'' if $R(s) = R'(s)$ for all $s \in \Sigma^*$, and ``no'' otherwise.

\medskip

\begin{theorem}
\label{tm-EQ}
The equivalence problem for RNNs is undecidable.
\end{theorem}

\begin{proof}
We prove by contradiction. Suppose Turing machine $M$ decides the equivalence problem. Given any deterministic Turing Machine $M'$, construct the RNN $R$ that simulates $M'$ on input $\epsilon$ as described in Corollary~\ref{cor:sim}. Let $E_{s,t}(a) = 0$ and $E_{s,t}(\$)=h_{s,t}(n_1)-h_{s,t}(n_2)$. If $M'$ does not halt on $\epsilon$, for all $t \in \mathbb{N}$, $E'_{s,t}(a) = E'_{s,t}(\$)=1/2$; if $M'$ halts after $T$ steps, $E'_{s,T}(a) = 1/(e+1)$, $E_{s,T}(\$)=e/(e+1)$. Let $R'$ be the trivial RNN that computes $\{a^n: P(a^n) = 2^{-(n+1)}, n \ge 0\}$. We run $M$ on input $\langle R, R'\rangle$. If $M$ returns ``no", $M'$ halts on $x$, else it does not halt. Therefore the Halting Problem would be decidable if equivalence is decidable. Therefore equivalence is undecidable.
\end{proof}

\section{Minimization}

We look next at minimization of RNNs.  For comparison, state-minimization of a deterministic PFSA is $O(|E|\log{|Q|})$ where $|E|$ is the number of transitions and $|Q|$ is the number of states \cite{aho_hopcroft_ullman_1974}. Minimization of a non-deterministic PFSA is PSPACE-complete \cite{jiang_ravikumar_1993}. 

We focus on minimizing the number of hidden neurons ($|N|$) in RNNs:

\paragraph{Minimization:} Given RNN $R$ and non-negative integer $n$, return ``yes'' if $\exists$ RNN $R'$ with number of hidden units $|N'| \le n$ such that $R(s) = R'(s)$ for all $s \in \Sigma^*$, and ``no'' otherwise.

\medskip

\begin{theorem}
\label{tm-MIN}
Minimization of RNNs is undecidable.
\end{theorem}

\begin{proof}
We reduce from the Halting Problem. Suppose Turing Machine $M$ decides the minimization problem. For any Turing Machine $M'$, construct the same RNN $R$ as in Theorem~\ref{tm-EQ}. We run $M$ on input $\langle R,0 \rangle$. Note that an RNN with no hidden unit can only output constant $E'_{s,t}$ for all $t$. Therefore the number of hidden units in $R$ can be minimized to $0$ if and only if it always outputs $E'_{s,t}(a) = E'_{s,t}(\$)=1/2$.  If $M$ returns ``yes", $M'$ does not halt on $\epsilon$, else it halts. Therefore minimization is undecidable.
\end{proof}
%\fi
\section{Conclusion}

We proved the following hardness results regarding RNN as a recognizer of weighted languages:

\begin{enumerate}
\item Consistency:
\begin{enumerate}
\item Inconsistent RNNs exist.
\item Consistency of RNNs is undecidable.
\end{enumerate}
\item Highest-weighted string:
\begin{enumerate}
\item Finding the highest-weighted string for an arbitrary RNN is undecidable.
\item Finding the highest-weighted string for a consistent RNN is decidable, but the solution length can surpass all computable bounds.
\item Restricting to solutions of polynomial length, finding the highest-weighted string is NP-complete and APX-hard.
\end{enumerate}
\item Testing equivalence of RNNs and minimizing the number of neurons in an RNN are both undecidable.
\end{enumerate}

Although our undecidability results are upshots of the Turing-completeness of RNN \cite{siegelmann}, our NP-completeness and APX-hardness results are original, and surprising, since the analogous hardness results in PFSA relies on the fact that there are multiple derivations for a single string \cite{casacuberta_higuera_2000}. The fact that these results hold for the relatively simple RNNs we used in this paper suggests that the case would be the same for more complicated models used in NLP, such as long short term memory networks (LSTMs; \citealt{Hochreiter:1997:LSM:1246443.1246450}).

Our results show the non-existence of (efficient) algorithms for interesting problems that researchers using RNN in natural language processing tasks may have hoped to find. On the other hand, the non-existence of such efficient or exact algorithms gives evidence for the necessity of approximation, greedy or heuristic algorithms to solve those problems in practice. In particular, since finding the highest-weighted string in RNN is the same as finding the most-likely translation in a sequence-to-sequence RNN decoder, our NP-completeness and APX-hardness results provide some justification for employing greedy and beam search algorithms in practice.

\section*{Acknowledgments}
This work was supported by DARPA (W911NF-15-1-0543 and HR0011-15-C-0115). Andreas Maletti was financially supported by DFG Graduiertenkolleg 1763 (QuantLA).

\bibliographystyle{acl_natbib}
\bibliography{rnn}

\clearpage
\appendix
\section*{Appendix}
\begin{customthm}{11}
Identifying the best string of polynomial length in a consistent RNN is NP-complete and APX-hard.
\end{customthm}

\begin{proof}
  Clearly, we can guess an input string of polynomial length, run the RNN, and verify whether its weight exceeds the given bound in polynomial time.  Therefore the problem is trivially in NP.  For NP-hardness we now reduce from the 0-1 Integer Linear Programming Feasibility Problem to our problem:

\paragraph{0-1 Integer Linear Programming Feasibility:} Given: $n$ variables $x_1, x_2, ..., x_n$ which can only take values in $\{0, 1\}$, and $k$ constraints ($k \ge 1$): $\sum_{j=1}^{n}{A_{ij}x_j} - B_i \le 0$, $\forall i \in [k]$. $A \in \mathbb{Z}^{k \times n}$, $B \in \mathbb{Z}^k$. Return: Yes iff. there is a feasible solution $x=(x_1, x_2, ..., x_n) \in \{0, 1\}^n$ that satisfies all $k$ constraints.
\medskip

Suppose we are given an instance of the above problem. Construct an instance of the consistent best string with polynomial length problem with input $\langle R, c\rangle$, where:

\begin{enumerate}
\item $R$ is an RNN as follows: 
\[\resizebox{0.45\textwidth}{!}{$\Sigma=\{0, 1\}, \{1,\dots,n,g_1, \dots, g_k,l_1, \dots, l_k\} \subset N$}\] 
$$\forall j \in [n], 0 \le t \le j-1: h_{s,t}(j) = 0$$ 
$$\forall j \in [n], t \ge j: h_{s,t}(j) = x_j$$ 
$$\forall i \in [k], 0 \le t \le n: h_{s,t}(g_i) = h_{s,t}(l_i) = 0$$ 
\[\resizebox{0.45\textwidth}{!}{$\forall i \in [k], t \ge n+1: h_{s,t}(g_i) = \sigma\langle 1-\sum_{j=1}^{n}{A_{ij}h_{s,t-1}(j)} +B_i\rangle$}\]
\[\resizebox{0.45\textwidth}{!}{$\forall i \in [k], t \ge n+1: h_{s,t}(l_i) = \sigma\langle-\sum_{j=1}^{n}{A_{ij}h_{s,t-1}(j)} +B_i\rangle$}\]
Let $d=\sum_{i=1}^{k}{(h_{s,t}(g_i)-h_{s,t}(l_i))}$, $\delta_2 = \frac{1}{k+2}$. We pick a big enough positive rational number $\beta$ so that if we define $\delta_1=\frac{1}{2e^{\beta}+1}$, \begin{equation} \label{eq:1}
\delta_1 < (\frac{1-\delta_1}{2})^n \delta_2
\end{equation} When $t \ne n+1$, set $$E_{s,t}(0) = E_{s,t}(1)=\beta, E_{s,t}(\$)=0.$$ Therefore $$E'_{s,t}(0) = E'_{s,t}(1)=\frac{1-\delta_1}{2}, E'_{s,t}(\$)=\delta_1$$ When $t = n+1$, one can verify that we can set $$E_{s,t}(0) = E_{s,t}(1)=\gamma_d, E_{s,t}(\$)=0$$ where $$\gamma_d = \ln{\frac{\frac{1}{(d+1)\delta_2}-1}{2}}$$ so that $$E'_{s,t}(\$)=(d+1)\delta_2$$ since the range of $d$ is a finite set of values $\{0, 1, 2, \dots, k\}$.
\item $c = (\frac{1-\delta_1}{2})^n(k\delta_2)$
\end{enumerate}
% If delta depends on n or k, then need to calculate the length of p and show it is polynomially bounded.

From equation \ref{eq:1} we get $$-\ln{\delta_1} > -n\ln{\frac{1-\delta_1}{2}}+\ln{(k+2)},$$ so we can pick $\beta$ such that its length written in binary \[\resizebox{0.48\textwidth}{!}{$\log_2{\beta}=\log_2{(\ln{(\frac{1}{\delta_1}-1)}-\ln{2})}=O(\log{(-\ln{\delta_1})})$}\] is logarithmic in $n$ and $k$. So the weights in matrices $E, E'$ that produce $\beta$ are polynomial in $n$ and $k$. Same is true for the weights that produce $\gamma_d$. $c$ written in binary has length $$-\log_2{c}=-n\log_2{\frac{1-\delta_1}{2}}-\log_2{(k\delta_2)}$$ $$=n-n\log_2{(1-\delta_1)}-\log_2{\frac{k}{k+2}}$$ which is polynomial in $n$ and $k$. So our construction is polynomial.

We now prove that if we can solve the $\langle R, c\rangle$-instance of consistent best string of polynomial length in polynomial time, we can also solve the given instance of 0-1 Integer Linear Programming Feasibility in polynomial time.

By our design, at time $1 \le t \le n$, $R$ reads a binary string $x \in \{0, 1\}^n$ into neurons $1,2,\dots, n$ while predicting almost half-half probability for either 0 or 1 and infinitesimal probability $\delta_1$ for termination. Therefore no string with length less then $n$ has weight greater than $c$.

At time $t = n+1$, since $\sum_{j=1}^{n}{A_{ij}h_{s,t-1}(j)}-B_i$ is an integer, $h_{s,t}(g_i) - h_{s,t}(l_i)$ is the indicator for whether the $i$-th constraint is satisfied: $$h_{s,t}(g_i) - h_{s,t}(l_i)=\left\{
    \begin{array}{c l}	
        0 & \sum_{j=1}^{n}{A_{ij}x_j} -B_i \ge 1\\
        1 & \sum_{j=1}^{n}{A_{ij}x_j} -B_i \le 0
    \end{array}\right.$$
Therefore $d$ is the total number of clauses satisfied by a given setting of $x=(x_1, x_2, \dots, x_n)$ ($0 \le d \le k$). The termination probability at $t = n+1$ is $(d+1)\delta_2 = \frac{d+1}{k+2}$. If all $k$ clauses are satisfied, this setting of $x$ would have termination probability $1-\delta_2$ and therefore weight $(\frac{1-\delta_1}{2})^n(1-\delta_2)>c$. If fewer than $k$ clauses are satisfied, $x$ would have weight at most $c$. 

When $t \ge n+2$, $R$ continues to assign almost half-half probability for either 0 or 1 and infinitesimal probability for termination. Therefore any string of length greater than $n+1$ has a weight smaller than $\epsilon$. From that point on the output vector is constant, so the RNN is consistent. Notice that the weights of strings monotonically decrease with length except for at length $n$.

Therefore our construction ensures that the only length at which a string can have weight greater than $c$ is $n$. Thus, if there is any string whose weight is greater than $c$, the given instance of 0-1 Integer Linear Programming Problem is feasible; otherwise it is not.

Define the maximum number of clauses satisfied by all assignments of $x \in \{0, 1\}^n$: $$d_{max} = \max_{x \in \{0, 1\}^n}{d(x)}.$$ By our construction, when $d_{max} \ge 1$, the highest-weighted string will occur at length $n$, and has weight $(\frac{1-\delta_1}{2})^n(d_{max}+1)\delta_2=(\frac{1-\delta_1}{2})^n\frac{d_{max}+1}{k+2}$ which is proportional to $d_{max}+1$. The empty string has the highest weight among all strings of length not equal to $n$. Its weight is $\delta_1 < (\frac{1-\delta_1}{2})^n \delta_2$ which will always be less than the weight of any length-$n$ string corresponding to a setting of variables satisfying at least 1 constraint ($ \ge (\frac{1-\delta_1}{2})^n(2\delta_2)$).

Therefore, given any rational number $\zeta = \frac{\eta_1}{\eta_2}> 0$ ($\eta_1, \eta_2 \in [k]$), define $\delta(\zeta) = \frac{\eta_1}{\eta_2+1}$. If binary string $s$ is a $(1-\delta(\zeta))$-approximation to consistent best string with polynomial length, then reading $s$ as a vector of $n$ variables $x = (x_1, x_2, \dots, x_n)$, $x$ would be a  $(1-\zeta)$-approximation to 0-1 Integer Linear Programming Maximum Satisfiability (the optimization version of the problem, i.e., finding a setting of variables to satisfy the greatest number of constraints).

Thus our reduction is PTAS-reduction and preserves approximability. Since 0-1 Integer Linear Programming Feasibility is NP-complete and APX-complete, consistent best string of polynomial length is NP-complete and APX-hard, meaning there is no Polynomial-Time Approximation Scheme to find the best string bounded by polynomial length (i.e. the best we can hope for in polynomial time is a constant-factor approximation algorithm) unless P $=$ NP.
\end{proof}

\begin{customlemma}{14}
  \label{lm:help}
  For every $k \in \natp$
  \[ \prod_{t \in \nat} \frac{e^{k(t+1)}}{e^{k(t+1)} + 1} = \frac
    2{(-1;\, e^{-k})_\infty} \enspace, \]
  where $(-1; e^{-k})_\infty$ is the infinite $e^{-k}$-Pochhammer
  symbol.
\end{customlemma}

\begin{proof}
  \begin{align*}
    &\phantom{{}={}} \prod_{t \in \nat} \frac{e^{k(t+1)}}{e^{k(t+1)} +
      1}
      = \prod_{t \in \natp} \Bigl(\frac{e^{kt}}{e^{kt} + 1} \cdot
      \frac{e^{-kt}}{e^{-kt}} \Bigr) \\
    &= \prod_{t \in \natp} \frac 1{1+e^{-kt}}
      = \Biggl(\Bigl(\prod_{t \in \natp} \frac 1{1+e^{-kt}}
      \Bigr)^{-1} \Biggr)^{-1} \\
    &= \Bigl(\prod_{t \in \natp} (1+e^{-kt}) \Bigr)^{-1}
      = \Bigl(\frac 12 \prod_{t \in \nat} (1+e^{-kt})
      \Bigr)^{-1} \\
    &= \frac 2 {(-1;\, e^{-k})_\infty} \tag*{\qedhere}
  \end{align*}
\end{proof}

\begin{customlemma}{15}
  \label{lm:incon}
  Reconsider the RNN of Example 3
  %~\ref{ex:consistent}
  and suppose that
  $h_{s, t}(n) - h_{s, t}(n') = 0$ for all $t \in \nat$.  Then
  \[ \sum_{s \in \Sigma^*} R(s) = 1 - \frac 2{(-1;\, e^{-2})_\infty}
    \approx 0.14 \]  
\end{customlemma}

\begin{proof}
  \begin{align*}
    &\phantom{{}={}} \sum_{s \in \Sigma^*} R(s)
      = \sum_{n \in \nat} R(a^n) \\
    &= \sum_{n \in \nat} \Bigl( \frac{e^{-(n+1)}}{e^{n+1} +
      e^{-(n+1)}} \cdot \prod_{t = 0}^{n-1} \frac{e^{t+1}}{e^{t+1} +
      e^{-(t+1)}} \Bigr) \\
    &= 1 - \prod_{t \in \nat} \frac{e^{2(t+1)}}{e^{2(t+1)} + 1}
      = 1 - \frac 2{(-1;\, e^{-2})_\infty} \\
    &\approx 0.14 \enspace,
  \end{align*}
  where the final equality utilizes Lemma~\ref{lm:help}.
\end{proof}

\begin{customlemma}{16}
  \label{lm:con}
  Reconsider the RNN of Example 3
  %~\ref{ex:consistent}
  and suppose that
  there exists a time point~$T \in \nat$ such that for all $t \in
  \nat$ 
  \[ h_{s, t}(n) - h_{s, t}(n') =
    \begin{cases}
      1 & \text{if } t = T \\
      0 & \text{otherwise.}
    \end{cases} \]
  Then
  \[ \sum_{s \in \Sigma^*} R(s) = 1 \]
\end{customlemma}

\begin{proof}
  \begin{align*}
    &\phantom{{}={}} \sum_{s \in \Sigma^*} R(s) = \sum_{n \in \nat}
      R(a^n) \\
    &= \Bigl(\sum_{n = 0}^T R(a^n) \Bigr) + \Bigl(\sum_{n =
      T+1}^\infty R(a^n) \Bigr) \\
    &= \sum_{n = 0}^T \Bigl( \frac {e^{-(n+1)}}{e^{n+1} + e^{-(n+1)}}
      \cdot \prod_{t = 0}^{n-1} \frac{e^{t+1}}{e^{t+1} + e^{-(t+1)}}
      \Bigr) \\
    &\phantom{{}={}} + \sum_{n = T+1}^\infty \frac
      {e^{2n-3T-4}}{e^{n+1} + e^{2n-3T-4}} \\
    &\phantom{{}={}} \qquad {} \cdot \Bigl(\prod_{t = 0}^T
      \frac{e^{t+1}}{e^{t+1} + e^{-(t+1)}} \Bigr) \\
    &\phantom{{}={}} \qquad {} \cdot \Bigl(\prod_{t = T+1}^{n-1}
      \frac{e^{t+1}}{e^{t+1} + e^{2t-3T-4}} \Bigr) \\
    &= \sum_{n = 0}^T \Bigl( \frac 1{e^{2(n+1)} + 1} \cdot \prod_{t =
      0}^{n-1} \frac{e^{2(t+1)}}{e^{2(t+1)} + 1} \Bigr) \\
    &\phantom{{}={}} + \sum_{n = T+1}^\infty \frac
      {e^{n-3T-5}}{1 + e^{n-3T-5}} \cdot \Bigl(\prod_{t = 0}^T
      \frac{e^{2(t+1)}}{e^{2(t+1)} + 1} \Bigr) \\
    &\phantom{{}={}} \qquad {} \cdot \Bigl(\prod_{t = T+1}^{n-1}
      \frac{1}{1 + e^{t-3T-5}} \Bigr) \\
    &= 1 - \Bigl(\prod_{t = 0}^T \frac{e^{2(t+1)}}{e^{2(t+1)} +
      1}\Bigr) \\
    &\phantom{{}={}} \quad {} + \Bigl(\prod_{t = 0}^T
      \frac{e^{2(t+1)}}{e^{2(t+1)} + 1} \Bigr) \\
    &\phantom{{}={}} \qquad {} \cdot \sum_{n =
      T+1}^\infty \frac {e^{n-3T-5}}{1 + e^{n-3T-5}} \cdot \Bigl(\prod_{t = T+1}^{n-1}
      \frac{1}{1 + e^{t-3T-5}} \Bigr) \\
    &= 1 - \Bigl(\prod_{t = 0}^T \frac{e^{2(t+1)}}{e^{2(t+1)} +
      1}\Bigr) \\
    &\phantom{{}={}} \quad {} \cdot \Biggl(1 - 1 + \prod_{t =
      T+1}^\infty \frac 1{1 + e^{t-3T-5}} \Biggr) \\ 
    &= 1 - \Bigl(\prod_{t = 0}^T \frac{e^{2(t+1)}}{e^{2(t+1)} +
      1}\Bigr) \cdot \Bigl(\prod_{t = T+1}^\infty \frac 1{1 +
      e^{t-3T-5}} \Bigr) \\
    &\geq 1 - \Bigl(\prod_{t = 0}^T \frac{e^{2(t+1)}}{e^{2(t+1)} + 1}
      \Bigr) \cdot \Bigl(\prod_{t \in \nat} \frac 1{1 + e^t} \Bigr) \\
    &= 1 \tag*{\qedhere}
  \end{align*}
\end{proof}

\begin{customlemma}{17}
  \label{lm:sim}
  
\end{customlemma}

\begin{proof}
  We set
$h_{s, -1}(n) = h_{-1}(n)$ for all $n \in N$ and $h'_{s, -1}(n') =
h'_{-1}(n')$ for all $n' \in N'$.  Then trivially $h'_{s, -1}(n_1) -
h'_{s, -1}(n_2) = h_{-1}(n) - 0 = h_{s, -1}(n)$.  Moreover,
\begin{align*}
  h'_{s,t}(n_1)
  &= \sigma \langle V \cdot h'_{s,t-1} + V'_{s[t]} \rangle(n_1) \\
  &= \sigma \langle \sum_{n' \in N'} V(n_1, n') \cdot h'_{s,t-1}(n') \\
  &+     V'_{s[t]}(n_1) \rangle \\
  &= \sigma \langle \sum_{n' \in N} \bigl( V(n_1, n'_1) \cdot
    h'_{s,t-1}(n'_1) \\ 
    &+ V(n_1, n'_2) \cdot h'_{s,t-1}(n'_2) \bigr) +
    V'_{s[t]}(n_1) \rangle \\
  &= \sigma \langle \sum_{n' \in N} U(n, n') \cdot \bigl(
    h'_{s,t-1}(n'_1) - h'_{s,t-1}(n'_2) \bigr) \\
    &+ U'_{s[t]}(n) \rangle \\ 
  &= \sigma \langle \sum_{n' \in N} U(n, n') \cdot h_{s, t-1}(n') + U'_{s[t]}(n) \rangle
\end{align*} 
Similarly, we can show that
\[ h'_{s,t}(n_2) = \sigma \langle \sum_{n' \in N} U(n, n') \cdot h_{s,
    t-1}(n') + U'_{s[t]}(n) - 1\rangle \]
Hence $h'_{s, t}(n_1) - h'_{s, t}(n_2) = h_{s, t}(n)$ as required.
\end{proof}
\end{document}